\documentclass[conference]{IEEEtran}

\usepackage{booktabs}
\usepackage[latin1]{inputenc}
\usepackage{graphicx}
\usepackage[cmex10]{amsmath}
\usepackage[margin=2cm,nohead]{geometry}
\usepackage{cite}
\usepackage{amsfonts}
\usepackage{amsmath}
\usepackage{amsthm}
\usepackage{mathrsfs,dsfont}
\usepackage{amssymb,mathrsfs}
\usepackage[mathscr]{euscript}
\usepackage{psfrag} 
\usepackage{epstopdf} 
\usepackage{color}
\usepackage{tikz}
\usepackage{caption}
\usepackage{subcaption}
\usepackage{pict2e}
\usepackage[labelsep=period]{caption}
\usepackage[font=small]{caption}
\usepackage{flushend}

\newtheorem{theorem}{Theorem}
\newtheorem{corollary}{Corollary}
\newtheorem{lemma}{Lemma}
\newtheorem{proposition}{Proposition}

\def\Iagg{I_\text{agg}}
\def\mIagg{\mu_{{\text{I}}_\text{agg}}}

\def\E{\mathbb{E}}
\def\plos{\mathcal{P}_\text{LoS}}
\def\pnlos{\mathcal{P}_\text{NLoS}}
\def\Pt{\text{P}_\text{t}}
\def\Lf{\text{L}_\text{f}}
\def\slos{\Psi_\text{LoS}}
\def\snlos{\Psi_\text{NLoS}}

\def\mlos{\mu_\text{LoS}}
\def\vlos{\sigma^2_\text{LoS}}
\def\mnlos{\mu_\text{NLoS}}
\def\vnlos{\sigma^2_\text{NLoS}}

\def\rc{\text{R}_\mathcal{C}}
\def\fA{\varphi_\text{A}}
\def\af{\text{K}_\text{f}}

\def\pcov{\mathcal{P}_\text{cov}}

\def\pri{\text{P}_{\text{r},i}}
\def\pr{\text{P}_\text{r}}

\def\Ac{\mathcal{A}_\text{c}}
\def\A{\mathcal{A}}
\def\pri{\text{P}_{\text{r},i}}
\def\pr{\text{P}_\text{r}}
\def\vc{\varphi_\text{c}}
\def\sinr{\mathsf{SINR}}
\def\mxi{\mu_\xi}
\def\vxi{\sigma^2_\xi}

\begin{document}
\title{Coverage Maximization for a Poisson Field of Drone Cells}

\author{Mohammad Mahdi Azari$^{1}$, Yuri Murillo$^{1}$, Osama Amin$^{2}$, Fernando Rosas$^{3,4}$ \\ \vspace{2mm}  Mohamed-Slim Alouini$^{2}$, Sofie Pollin$^{1}$ \\ 
$^1$Department of Electrical Engineering, KU Leuven, Belgium \\ 
$^2$CEMSE Division, King Abdullah University of Science and Technology, Saudi Arabia \\ $^{3}$ Centre of Complexity Science and Department of Mathematics, Imperial College London, UK \\
$^{4}$ Department of Electrical and Electronic Engineering, Imperial College London, UK \vspace{2mm} \\
 Email: mahdi.azari@kuleuven.be}
\maketitle

\begin{abstract}
The use of drone base stations to provide wireless connectivity for ground terminals is becoming a promising part of future technologies. The design of such aerial networks is however different compared to cellular 2D networks, as antennas from the drones are looking down, and the channel model becomes height-dependent. In this paper, we study the effect of antenna patterns and height-dependent shadowing. We consider a random network topology to capture the effect of dynamic changes of the flying base stations. First we characterize the aggregate interference imposed by the co-channel neighboring drones. Then we derive the link coverage probability between a ground user and its associated drone base station. The result is used to obtain the optimum system parameters in terms of drones antenna beamwidth, density and altitude. We also derive the average LoS probability of the associated drone and show that it is a good approximation and simplification of the coverage probability in low altitudes up to 500 m according to the required signal-to-interference-plus-noise ratio ($\sinr$).
\end{abstract}
\vspace{3mm}


\section{Introduction}

The demand of efficient ubiquitous high-speed communication networks is an essential requirement for future communications. Deploying new wireless networks faces different challenges such as spectrum scarcity, fixed location and the need to set up new infrastructure \cite{agiwal2016next}. Another challenge is providing wireless connectivity to ground terminals when the existing terrestrial networks fail to operate or satisfy the demand of wireless connections \cite{miranda2016survey}. The use of drone base stations is an alternative future technology that can provide wireless connectivity for ground users. Furthermore, the swiftness of drones deployments in an independent fashion of the legacy infrastructure highlights the privilege of running such systems. To this end, NASA is prototyping a drones traffic management technology to facilitate the use of multiple drones in a range of altitudes \cite{kopardekar2014unmanned}.


Recently, people from research and academia highlighted the important performance benefits achieved by the use of drones as base stations.
The drone base stations can deploy a flexible scalable network with robust links due to the high line-of-sight (LoS) probability between the drone and a ground terminal. The existing research analyzed the performance gain achieved from this on-the-fly solution assuming single drones and multiple drones \cite{azari2016optimal,mozaffari2017mobile,azari2016gcw,azari2017ultra, mozaffari2016efficient,hayajneh2016drone}. In the single drone research studies, the altitude of a drone is optimized to balance between power and coverage requirements in \cite{azari2016gcw,azari2017ultra}. 
The performance of multiple drone base stations is analyzed in \cite{mozaffari2016efficient,hayajneh2016drone}, where the  effect of interference from neighboring drone base stations is analyzed. In \cite{mozaffari2016efficient}, multiple drone base stations are considered in a predefined topology to provide coverage for the ground terminals in downlink scenario, however only the interference coming from the nearest drone is taken into account. In \cite{hayajneh2016drone}, the coverage probability for a fixed number of drones is analyzed without considering the effect of randomness in the number of drones, drones antenna pattern and the height-dependent shadowing. 

In fact, most of the existing literature focus on the fixed network topology of drone base stations meaning that the number and location of the available drones are known a priori and remain unchanged. However, this snapshot is not always the same in practice due to the intrinsic mobile nature of the drones. Specifically, the continuing dynamic nature of this on-the-fly network results from different aspects such as smart design \cite{fotouhi2016dynamic,fotouhi2017dynamic}, limited flying lifetime \cite{chandrasekharan2016designing} and traffic management \cite{kopardekar2014unmanned}. The drone network can be smartly designed to benefit from the drones mobility by tunning  their positions according to the quality of service requirements \cite{fotouhi2016dynamic,fotouhi2017dynamic}. As for the energy consumption, it depends on the drones trajectory of movement, the duration of communication, payload weight and the battery size, which allow a limited range of lifetime \cite{chandrasekharan2016designing}. Therefore, the lifetime of a drone is random and compensation drones may be needed to support the service requirement of ground terminals. Taking these facts into account, the location and number of drone base stations are likely to change to allow other drones maneuver~\cite{kopardekar2014unmanned}.

In this paper, we study the coverage probability performance of multiple drone base stations while considering flexible dynamic changes of this scalable on-the-fly network. To this end, and without a priori knowledge regarding the exact number and location of the available drones, we model the distribution of the drones by a Poisson point process (PPP). Moreover, we assume that the drones employ a directional antenna to reduce the aggregate interference and concentrate on the target regions. By taking into account the height-dependent shadowing effect, we characterize the aggregate interference and derive the link coverage probability which includes the impact of different system parameters such as drones antenna pattern. Then, we obtain an average LoS probability of the associated drone to a target ground terminal which gives us an insight into the existence of the optimum drones antenna beamwidth, density and altitude. 
We also show that the coverage probability is approximated well by the average LoS probability up to 500 m dependent on the signal-to-interference-plus-noise ratio ($\sinr$) requirement. This approximation significantly reduces the complexity of the coverage probability expression.  

The rest of this paper is organized as follows. In Section II the network model and drone association strategy is discussed. The network performance is analyzed in Section III. Section IV presents the numerical results. Finally, we conclude the paper in Section V.

\section{Network Model}

In the following we describe the system structure in Section \ref{sys_architecture}, the channel model in Section \ref{channel_model} and the drone association method in Section \ref{drone_asso}.

\begin{figure}  
\centering
  \begin{tikzpicture}
  \node[above right] (img) at (0,0)   {\includegraphics[width=\linewidth]{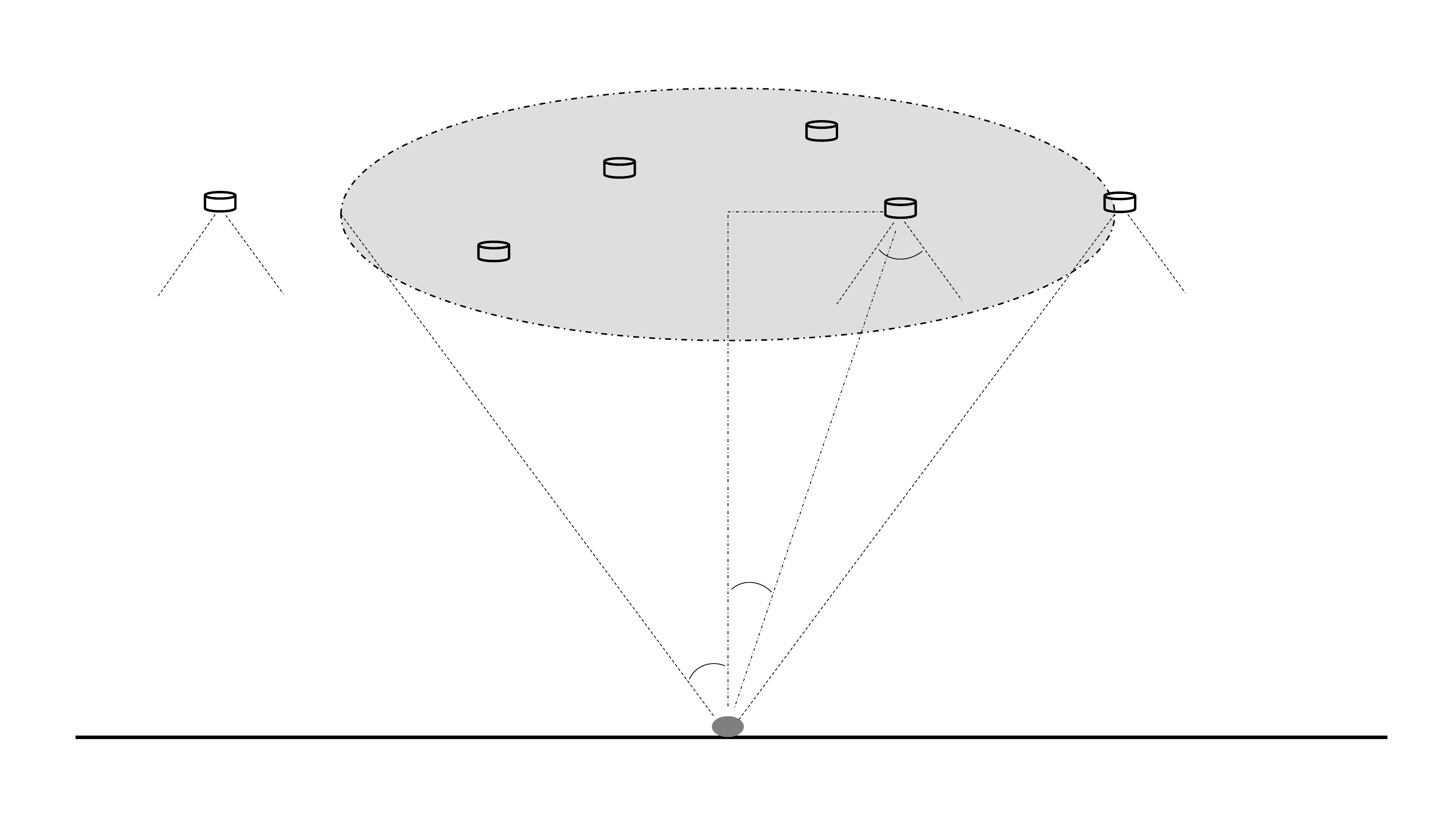}};
  \node at (170pt,110pt) {\footnotesize Drone};
  \node at (125pt,10pt) {\footnotesize UE};
  \node at (100pt,80pt) {\footnotesize $\A$};
  \node at (130pt,47pt) {\footnotesize $\varphi$};
  \node at (120pt,35pt) {\footnotesize $\frac{\fA}{2}$};
  \node at (125pt,110pt) {\footnotesize \text{O}};
  \node at (123pt,70pt) {\footnotesize $h$};
  \node at (140pt,100pt) {\footnotesize $r$};
  \node at (155pt,93pt) {\footnotesize $\fA$};
  \node at (147pt,70pt) {\footnotesize $d$};
  \end{tikzpicture}
   \caption{Downlink system between a drone and a ground UE. The connection is established with the closest drone and the rest conform a random field of interferers.} \label{sysmodel}
\end{figure}
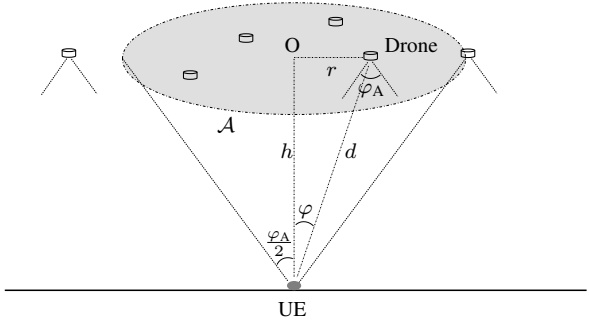

\subsection{System Architecture} \label{sys_architecture}

We consider a downlink communication system where ground user equipments (UEs) are served by drones acting as aerial base stations and providing wireless connectivity. We assume that the drone base stations are distributed according to a PPP of a fixed density $\lambda$ and placed at a same altitude $h$. We also assume the drones employ directional antennas with the same beamwidth $\fA$ pointed towards the ground whereas the ground UEs are equipped with omni-directional antennas. From the geometry of the network, the drones that only placed within a circular surface $\A$ centered at $O$ above a UE will be able to reach this device within their main lobes. This region is seen by the UE with the angle of $\fA$ as illustrated in Figure \ref{sysmodel}. The distance from O to a drone and the link length between the UE and the drone are denoted respectively as $r$ and $d$, while $\varphi$ is the complement of the elevation angle that the drone makes with respect to the UE.

\subsection{Channel Model} \label{channel_model}

In order to model the wireless channel between a ground UE and a drone, the LoS and non-line-of-sight (NLoS) components are considered separately along with their probabilities of occurrence \cite{al2014modeling}. Following this model and considering that all drones transmit at the same power level $\Pt$, the received power at the UE from the LoS and NLoS components can be expressed as 
\begin{equation}\label{Pr}
\ {P}_\text{r} =
 \begin{cases} 
      \frac{\Pt}{\Lf\slos} & ;~ \text{for LoS} \\
      \frac{\Pt}{\Lf\snlos} & ;~ \text{for NLoS}
   \end{cases}
\end{equation}
where $\Lf$ represents the free-space path loss (FSPL) and $\slos$ and $\snlos$ account for the excessive path loss and shadowing effects. Let us express the FSPL as
\begin{equation}\label{Lf}
\Lf = \left(\frac{4\pi fd}{c}\right)^2 = \af~\frac{h^2}{\cos^2(\varphi)}
\end{equation}
with $f$ being the frequency of operation, $c$ the speed of light, $\af=(4\pi f/c)^2$ for convenience of notation and finally \mbox{$d = \frac{h}{\cos(\varphi)}$} from Figure \ref{sysmodel}.
The terms $\slos$ and $\snlos$ follow a log-normal distribution 
\begin{align} \label{psi_distr}
10\log_{10}\Psi_\xi \sim \mathcal{N}(\mxi,\vxi);~~~ \xi \in \{\text{LoS},\text{NLoS}\} 
\end{align}
where $\mlos$ and $\mnlos$ describe the excessive path loss and are constant values depending on $f$ and the propagation environment. As shown in \cite{al2014modeling}, $\sigma_\text{LoS}$ and $\sigma_\text{NLoS}$ describe the shadow fading in the links and can be expressed as
\begin{align}
\sigma_\xi = a_\xi \cdot e^{b_\xi \varphi};~~~ \xi \in \{\text{LoS},\text{NLoS}\}
\end{align}
with parameters $a_\text{LoS}$, $b_\text{LoS}$, $a_\text{NLoS}$ and $b_\text{NLoS}$ being frequency and environment dependent.

Moreover, the LoS probability is given by \cite{al2014modeling}
\begin{equation} \label{Plos}
\plos(\varphi) = \beta_1\left(\frac{5\pi}{12}-\varphi\right)^{\beta_2}
\end{equation}
where $\beta_1$ and $\beta_2$ are also frequency and environment dependent parameters and the NLoS probability is \mbox{$\pnlos = 1 - \plos$.} 

\subsection{Drone Association and Link $\sinr$} \label{drone_asso}

The same as a regular cellular networking, we assume that a UE connects to the closest drone base station \cite{elsawy2016modeling}, where its angle with respect to the UE is a random variable represented by $\Phi_\text{c}$. Therefore, the connection link between the UE and the associated drone is interfered by all the other drone base stations. However, based on a combination of two reasons the effect of interfering drones outside of $\A$ can be neglected. The first reason is that for a fixed altitude $h$, increasing $\varphi$ decreases the LoS probability \cite{al2014modeling} and also increases the link length $d$ which lead to a relatively higher path loss. On the other hand, a directional antenna has a much lower gain outside its main lobe. Therefore, the received power from the interfering drones beyond $\A$ is significantly lower compared to the ones within the region.

Using the described channel model and considering $\pri$ as the received power from the $i$th interferer within $\A$, the aggregate interference can be written as
\begin{equation} \label{Iagg}
\Iagg = \sum_{i \in \mathcal{I}} \pri,
\end{equation}
where $\mathcal{I}$ indicates the set of interferers within $\A$. In \eqref{Iagg} $\Iagg$ becomes a stochastic process due to the random nature of the contributions involved. Indeed, for every realization of the PPP the number, location and channel statistics of the interferers will be different. To the best of our knowledge, there is no closed-form expression for the probability density function (pdf) of $\Iagg$ for such system model described above. For this reason and in order to provide a tractable analysis for the network performance we characterize the aggregate interference $\Iagg$ with its mean value $\mIagg$. This will allow us to obtain closed-form expressions for the performance metrics of the considered system and investigate the impact of system parameters on the network performance. In order to do so, we are interested in the signal-to-interference-plus-noise ratio ($\sinr$) of the communication link between a UE and its associated drone which can be written as follows
\begin{equation}\label{SINR}
\mathsf{SINR} = 
 \begin{cases} 
      \frac{\Pt}{(\mIagg+N_0) \Lf\slos} & ;~ \text{for LoS} \\
      \frac{\Pt}{(\mIagg+N_0) \Lf\snlos} & ;~ \text{for NLoS}
   \end{cases}
\end{equation}
where $N_0$ is the noise power.

\section{Performance Analysis}

In this section we elaborate on the network performance by adopting coverage probability as the performance metric. The coverage probability $\pcov$ of the link between a ground terminal and its associated drone is defined as
\begin{equation} \label{pcov}
\pcov \triangleq \mathbb{P}[\mathsf{SINR}>\text{T}],
\end{equation}
where $\mathbb{P}[E]$ is the probability of the event $E$, T is an $\mathsf{SINR}$ threshold and $\sinr$ is expressed in \eqref{SINR}. To compute $\pcov$, first we propose the following lemma which returns the mean aggregate interference.
\begin{lemma} \label{mIagg-lemma}
Given the relative location of the closest drone to a UE at \mbox{$\Phi_\text{c} = \varphi_\text{c}$}, the mean aggregate interference $\mIagg(\varphi_\text{c})$ is obtained as
\begin{equation} \label{mIagg-final}
\mIagg(\varphi_\text{c}) = \frac{2\pi \lambda \Pt}{\af} \cdot \Upsilon_{\!\mathrm{I}}(\varphi_\text{c}),
\end{equation}
where
\begin{align} \nonumber
\Upsilon_{\!\mathrm{I}} &(\varphi_\text{c}) = \int_{\varphi_\text{c}}^{\frac{\fA}{2}} \tan(\varphi) \left[\beta_1\left(\frac{5\pi}{12}-\varphi\right)^{\beta_2} 10^{\frac{-\mlos+v\vlos(\varphi)/2}{10}} \right. \\ \label{mIaggfinal}
&\left. + \left(1-\beta_1\left(\frac{5\pi}{12}-\varphi\right)^{\beta_2}\right) 10^{\frac{-\mnlos+v\vnlos(\varphi)/2}{10}} \right] d\varphi,
\end{align}
and $v=\frac{\ln(10)}{10}$.
\end{lemma}
\begin{proof}
The proof can be found in Appendix \ref{mIagg_proof}.
\end{proof}
Note that the mean aggregate interference $\mIagg(\vc)$ is independent from $h$. To elaborate this, we note that from PPP assumption the average number of interferers is equal to $\lambda \pi h^2 [\tan^2(\fA / 2)-\tan^2(\vc)]$, which increases proportional to $h^2$ as $h$ increases. On the other hand, using \eqref{Pr} and \eqref{Lf} the received power from each interferer decreases with the same rate. Therefore, the average cumulative interference remains the same at different altitudes. Moreover, the expression \eqref{mIagg-final} shows that $\mIagg(\vc)$ increases linearly with $\lambda$ due to the fact that the average number of interferers linearly increases with $\lambda$. Also $\mIagg(\vc)$ is higher for larger $\fA$ as the number of drones within $\A$ increases with $\fA$ for a fixed $h$ and $\lambda$.

Now using Lemma \ref{mIagg-lemma} the coverage probability can be found in the following theorem.

\begin{theorem} \label{pcov-theorem}
The coverage probability $\pcov$ of the communication link between a ground terminal and its associated drone can be expressed as
\begin{subequations} \label{pcovANDothers}
\begin{equation} \label{pcov-final}
\pcov(h,\fA,\lambda) = 2\pi \lambda h^2 \cdot \Upsilon_{\!\mathrm{cov}}(h,\fA,\lambda),
\end{equation}
where
\begin{align} \nonumber
\Upsilon_{\!\mathrm{cov}}(h,\fA,\lambda)& =\int_0^{\frac{\fA}{2}} \left[Q\left(\frac{\mlos-\psi(\varphi_\text{c})}{\sigma_\text{LoS}(\varphi_\text{c})}\right)\beta_1\left(\frac{5\pi}{12}-\vc \right)^{\beta_2} \right. \\ \nonumber
 & \left. \hspace{-10mm} + Q\left(\frac{\mnlos-\psi(\varphi_\text{c})}{\sigma_\text{NLoS}(\varphi_\text{c})}\right)\left(1-\beta_1\left(\frac{5\pi}{12}-\vc \right)^{\beta_2}\right)\right] \\ \label{pcov-integral}
& \times \frac{\sin(\varphi_\text{c})}{\cos^3(\varphi_\text{c})} e^{-\lambda \pi h^2 \tan^2(\varphi_\text{c})} ~d\varphi_\text{c}.
\end{align}
and
\begin{equation} \label{psi}
\psi(\varphi_\text{c}) = 10\log_{10}\left(\frac{\Pt}{[\mIagg(\vc) + N_0]\Lf(\varphi_\text{c})\text{T}}\right).
\end{equation}
\end{subequations}
\end{theorem}
\begin{proof}
The coverage probability in \eqref{pcov} can be written as
\begin{equation}\label{pcov-calculation}
\pcov(h,\fA,\lambda) = \int_0^{\frac{\fA}{2}} \mathbb{P}[\mathsf{SINR}>\text{T}|\Phi_\text{c} = \varphi_\text{c}] ~f_{\Phi_\text{c}}(\varphi_\text{c}) d\varphi_\text{c},
\end{equation}
where $f_{\Phi_\text{c}}(\varphi_\text{c})$ is the pdf of $\Phi_\text{c}$. Using the auxiliary random variable $R_\text{c}$ which represents the radius of the closest drone over $\A$ one finds
\begin{align} \nonumber
F_{\Phi_\text{c}}(\varphi_\text{c}) & \triangleq \mathbb{P} [\Phi_\text{c} \leq \varphi_\text{c}] =  \mathbb{P} [R_\text{c} \leq h \tan(\varphi_\text{c})] \\ \label{cdf-closest}
& = 1-e^{-\lambda \pi h^2 \tan^2(\varphi_\text{c})},
\end{align}
where the last equation comes from the fact that the null probability of the PPP in an area $\mathcal{C}$ is $\exp(-\lambda |\mathcal{C}|)$. Using \eqref{cdf-closest} we have
\begin{equation}\label{pdf-closest}
f_{\Phi_\text{c}}(\varphi_\text{c}) = \frac{\partial}{\partial \varphi_\text{c}} F_{\Phi_\text{c}}(\varphi_\text{c}) = 2\pi \lambda h^2 \frac{\sin(\varphi_\text{c})}{\cos^3(\varphi_\text{c})} e^{-\lambda \pi h^2 \tan^2(\varphi_\text{c})}.
\end{equation}
On the other hand, one obtains
\begin{subequations}
\begin{align} \nonumber
\mathbb{P}&[\mathsf{SINR}>\text{T}|\Phi_\text{c} = \varphi_\text{c}] \\ \nonumber 
& = \mathbb{P}[\mathsf{SINR}_\text{LoS}>\text{T}|\Phi_\text{c} = \varphi_\text{c}]\cdot\plos(\varphi_\text{c}) \\ \nonumber 
& + \mathbb{P}[\mathsf{SINR}_\text{NLoS}>\text{T}|\Phi_\text{c} = \varphi_\text{c}]\cdot\pnlos(\varphi_\text{c}) \\ \nonumber
& = \mathbb{P}\left[\slos<\frac{\Pt}{[N_0 + \mIagg(\varphi_\text{c})]\Lf(\varphi_\text{c})\text{T}}\right]\cdot\plos(\varphi_\text{c}) \\ \label{Pr-sir1}
& + \mathbb{P}\left[\snlos<\frac{\Pt}{[N_0 + \mIagg(\varphi_\text{c})]\Lf(\varphi_\text{c})\text{T}}\right]\cdot\pnlos(\varphi_\text{c}) \\ \nonumber
& = Q\left(\frac{\mlos-\psi(\varphi_\text{c})}{\sigma_\text{LoS}(\varphi_\text{c})}\right)\beta_1\left(\frac{5\pi}{12}-\vc \right)^{\beta_2} \\ \label{Pr-sir2}
& + Q\left(\frac{\mnlos-\psi(\varphi_\text{c})}{\sigma_\text{NLoS}(\varphi_\text{c})}\right)\left(1-\beta_1\left(\frac{5\pi}{12}-\vc \right)^{\beta_2}\right).
\end{align}
\end{subequations}
Above, in \eqref{Pr-sir1} the equation in \eqref{SINR} is used, \eqref{Pr-sir2} follows from \eqref{psi_distr} and \eqref{Plos}, and $\psi(\varphi_\text{c})$ is stated in \eqref{psi}. Finally by using \eqref{psi}, \eqref{pcov-calculation}, \eqref{pdf-closest}, \eqref{Pr-sir2} and Lemma \ref{mIagg-lemma} the desired result is attained.
\end{proof}
According to Theorem \ref{pcov-theorem}, we observe first that $\psi(\vc)$ is a decreasing function of $\lambda$ and $h$ since $\mIagg(\vc)$ and $\Lf$ are respectively increasing function of $\lambda$ and $h$ in \eqref{psi}. Since the Q--function is a decreasing function, and the exponential term in \eqref{pcov-integral} decreases with both $\lambda$ and $h$, the function $\Upsilon_{\!\mathrm{cov}}$  decreases with an increase in both $\lambda$ and $h$. In contrary, the first term in \eqref{pcov-final} increases with the increase in $\lambda$ and $h$. Therefore, these two opposite contributors finally suggest the existence of an optimum $\lambda$ and $h$ for maximum coverage probability which are numerically observed in the next section. 

To investigate the impact of $\fA$ on $\pcov$, and since the coverage probability is directly influenced by the LoS probability, we derive the average LoS probability of a closest drone $\bar{\mathcal{P}}_\text{LoS}^{\hspace{1pt}\text{c}}$ in the following proposition.

\begin{proposition} \label{PcovApp}
The average LoS probability of the closest drone $\bar{\mathcal{P}}_\text{LoS}^{\hspace{1pt}\text{c}}$  is given by
\begin{equation} \label{PlosC-proposition}
\bar{\mathcal{P}}_\text{LoS}^{\hspace{1pt}\text{c}} (h,\fA,\lambda) = 2\pi \lambda h^2 \cdot \Upsilon_{\!\mathrm{LoS}} (h,\fA,\lambda),
\end{equation}
where 
\begin{equation} \label{UpsilonLoS}
\Upsilon_{\!\mathrm{LoS}} = \int_0^{\frac{\fA}{2}} \frac{\sin(\varphi)}{\cos^3(\varphi)} e^{-\lambda \pi h^2 \tan^2(\varphi)} \beta_1  \left(\frac{5\pi}{12}-\varphi \right)^{\beta_2} \!\! d\varphi
\end{equation}
\end{proposition}
\begin{proof}
The proof is given in Appendix \ref{PcovApp_proof}.
\end{proof}

In the next section, we numerically show that $\bar{\mathcal{P}}_\text{LoS}^{\hspace{1pt}\text{c}}$ given in Proposition \ref{PcovApp} is a good approximation of $\pcov$ expressed in Theorem \ref{pcov-theorem} at low altitudes. Note that the coverage probability expression in Theorem \ref{pcov-theorem} includes two integrals ($\mIagg$ in \eqref{psi} obtained from Lemma \ref{mIagg-lemma} includes one integral), however its approximation $\bar{\mathcal{P}}_\text{LoS}^{\hspace{1pt}\text{c}}$ only requires one integral calculation.

\begin{corollary} \label{Plos-increasing}
The average LoS probability of the closest drone $\bar{\mathcal{P}}_\text{LoS}^{\hspace{1pt}\text{c}}$ is an increasing function of $\fA$.
\end{corollary}
\begin{proof}
The term inside of the integral in \eqref{UpsilonLoS} is positive. Therefore, increasing $\fA$ leads to a bigger $\Upsilon_{\!\mathrm{LoS}}$ and hence a higher $\bar{\mathcal{P}}_\text{LoS}^{\hspace{1pt}\text{c}}$ in \eqref{PlosC-proposition}. 
\end{proof}

Corollary \ref{Plos-increasing} provides an insight into the existence of an optimum value for $\fA$. In fact, $\bar{\mathcal{P}}_\text{LoS}^{\hspace{1pt}\text{c}}$ and hence the coverage probability increases with an increase in $\fA$, however using Lemma \ref{mIagg-lemma} the interference also increases with $\fA$, reducing the coverage probability. Therefore, at some $\fA$ one expects that the influence of these two contributors are balanced and the coverage probability reaches its maximum. 

\section{Numerical Results}
This section provides numerical results for the considered downlink  system, in which the optimal configuration of $h$, $\lambda$ and $\fA$ is found. Unless otherwise stated, the following system parameters are used in the numerical results: \mbox{$\Pt = -6$ dB}, \mbox{$N_0 = -150$ dB}, $\lambda = 5\times 10^{-6}$, $\fA = 90^{\circ}$, $h=500$ m, T = -5 dB, \mbox{$f = 2$ GHz}, and Urban environment.

\begin{figure}  
\centering
\includegraphics[width=\linewidth]{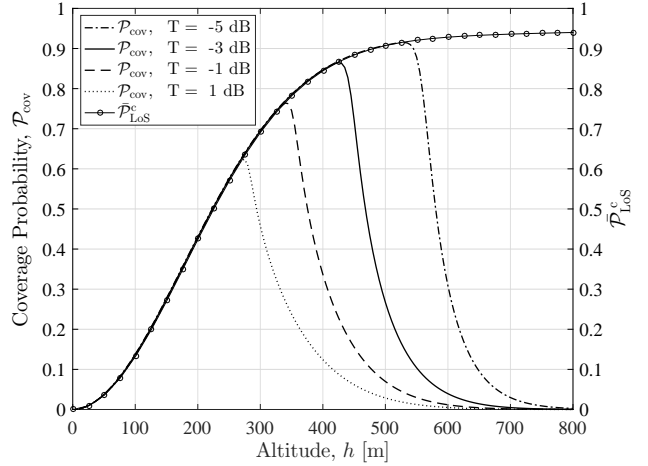}
   \caption{The coverage probability $\pcov$ is well approximated by the average LoS probability of the closest drone $\bar{\mathcal{P}}_\text{LoS}^{\hspace{1pt}\text{c}}$ for low altitudes. Higher altitudes translate into longer link lengths but fixed mean aggregate interference, so reaching the $\sinr$ threshold becomes more difficult and the two curves deviate from each other.} \label{Pcov_h_T}
\end{figure} 

First, we study the coverage probability performance versus the altitude for different $\sinr$ thresholds in Figure \ref{Pcov_h_T}. Moreover, we study the behavior of the average LoS probability of the closest drone $\bar{\mathcal{P}}_\text{LoS}^{\hspace{1pt}\text{c}}$ and compare it with the coverage probability performance. We observe that  $\bar{\mathcal{P}}_\text{LoS}^{\hspace{1pt}\text{c}}$ is a good approximation of  $\pcov$ at low altitudes. As the altitude increases so does the LoS probability of the closest drone, resulting in a better channel condition. However, the link length becomes longer while the aggregate interference remains the same, decreasing the $\sinr$. These two effects balance each other at an optimum altitude, which maximizes the value of $\pcov$. More restrictive values of the $\sinr$ threshold T translate into a faster decay of the coverage probability with altitude, limiting the range in which the approximation is valid.

\begin{figure}  
\centering
\includegraphics[width=\linewidth]{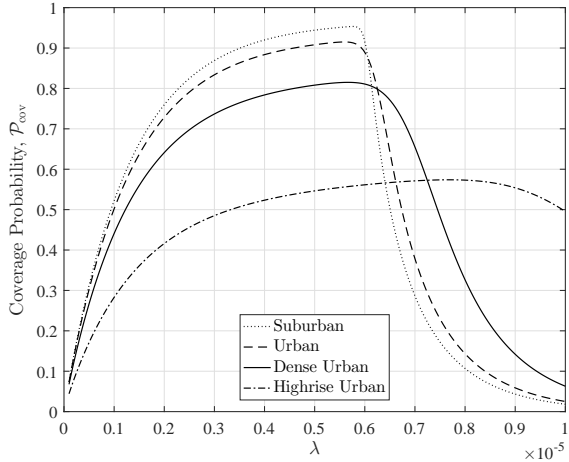}
   \caption{The coverage probability increases with $\lambda$ until the number of interferers becomes too large and the $\sinr$ decreases. This effect is lower in more crowded environments due to higher NLoS probability, resulting in lower aggregate interference.} \label{Pcov_lambda_env}
\end{figure} 

Figure \ref{Pcov_lambda_env} depicts the existence of the optimum density $\lambda$ that maximizes the coverage probability. When the density of nodes is low the probability of the closest drone being placed closer to the UE is also low, resulting in a worse $\sinr$ and therefore lower coverage probability. However, this probability quickly increases with $\lambda$. This behavior is maintained until the number of interferers becomes too large and the $\sinr$ starts to degrade. When comparing different environments, one can see that more crowded scenarios result in lower coverage probability due to lower LoS probability of the links. On the other hand, this effect is beneficial when the number of drones becomes large, as the received aggregate interference is also lower.

\begin{figure}  
\centering
\includegraphics[width=\linewidth]{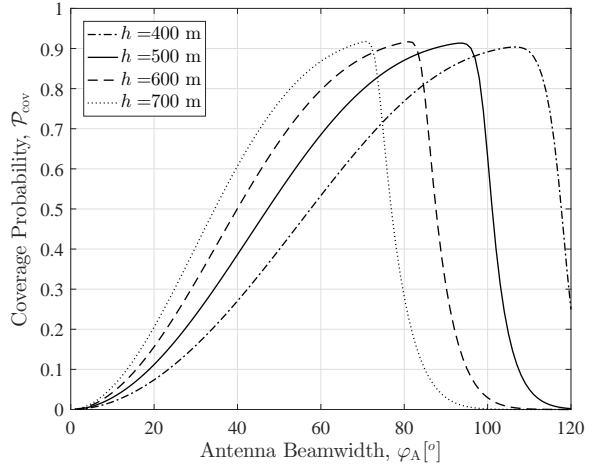}
   \caption{Increasing the antenna beamwidth increases the probability of the closest drone reaching the UE, but the same applies to the rest of interferers.} \label{Pcov_fA_h}
\end{figure} 

Finally, we investigate the impact of antenna beamwidth on the coverage probability for different altitude in Figure \ref{Pcov_fA_h}. When $\fA$ is narrow, the closest drone may not be able to reach the UE. As $\fA$ increases, the received power increases improving the coverage probability. With the further increase of $\fA$ the aggregate interference increases since all drones have the same beamwidth, causing a degradation of  the coverage probability. This behavior of the coverage probability shows the existence of an optimum antenna beamwidth $\fA$ for different altitudes.  

\section{Conclusion}

This paper considers a Poisson field of drone base stations that employ identical directional antennas. By associating the closest drone to a ground UE and considering the other drones as interferers, we derived the coverage probability of the corresponding link. Our analysis showed that the system parameters including the density of the drones, the altitude of the flying base stations and the beamwidth of their antennas influence the coverage probability significantly. Moreover, the analysis and numerical results proved that the coverage probability of such system can be well approximated by the average LoS probability of the closest drone at low altitudes up to 500 m. This approximation simplifies the large expression for the coverage probability and gives us an insight into the critical trading-off points of the system. In the future, we will consider the drone base stations co-existing with terrestrial networks and derive the coverage probability and optimum system parameters for such hybrid networks. Moreover, we will generalize the results for different network topologies using the stochastic tools.

\begin{appendices}

\section{Proof of Lemma \ref{mIagg-lemma}} \label{mIagg_proof}

To characterize the mean aggregate interference we notice that given the location of the closest drone at $\Phi_\text{c} = \varphi_\text{c}$ over $\A$ the interfering drones are distributed according to PPP of the fixed density $\lambda$ over a sub-region $\Ac$ characterized by $\vc \leq \varphi \leq \frac{\fA}{2}$ \cite{andrews2011tractable,naghshin2017coverage}. Therefore, by assuming that $N$, $\Phi$ and $\Psi$ are random variables representing the number, location and channel statistic of the interferers, one sees
\begin{align} \nonumber
\mIagg(\varphi_\text{c}) &= \E[\Iagg(\vc)] = \E_{N,\Phi,\Psi}\left[\sum_{i=1}^{N}\pri \right] \\ \nonumber
 &= \E_{N}\left[\E_{\Phi,\Psi}\left(\sum_{i=1}^{n}\pri \Big{|} N = n\right) \right] \\ \label{mIagg}
 &= \E_{N}\left[\sum_{i=1}^{n} \E_{\Phi,\Psi}[\pri] \Big{|} N = n \right].
\end{align}
The received interfering powers $\pri$ are identical and assumed to be independent random variables (i.i.d.) 
with a same distribution denoted as $\pr$. On the other hand, from the property of PPP $N$ follows a Poisson distribution of the mean value $\lambda |\Ac|$. Therefore \eqref{mIagg} can be re-written as
\begin{align}\label{mIagg2}
\mIagg(\vc) &= \E_{N}\left[n~\E_{\Phi,\Psi}[\pr] \right] 
  = \lambda |\Ac| \cdot \E_{\Phi,\Psi}[\pr].
\end{align}
To compute $\E_{\Phi,\Psi}[\pr]$, first we evaluate the pdf of $\Phi$ denoted as $g_\Phi(\varphi)$. To this end we note that
\begin{equation} \nonumber
G_\Phi(\varphi) \triangleq \mathbb{P}[\Phi \leq \varphi] = \frac{\pi (r^2-r_\text{c}^2)}{|\Ac|};~~~\vc<\varphi<\frac{\fA}{2}
\end{equation}
where $r_\text{c}$ corresponds to the closest drone and $r = h\tan(\varphi)$. Thus
\begin{equation} \label{gPhi}
g_\Phi(\varphi) = \frac{\partial}{\partial \varphi} G_\Phi(\varphi) = \frac{\partial r}{\partial \varphi} \cdot \frac{\partial G_\Phi(\varphi)}{\partial r}  = \frac{2\pi h^2}{|\Ac|} \frac{\sin(\varphi)}{\cos^3(\varphi)}.
\end{equation} 
Using \eqref{gPhi} one can write
\begin{align} \nonumber
\E_{\Phi,\Psi}[\pr] = \int_{\vc}^{\frac{\fA}{2}} \E_{\Psi}[\pr|\Phi = \varphi]~ g_\Phi(\varphi)~ d\varphi \\ \label{Ephi_psi}
= \frac{2\pi h^2}{|\Ac|}\int_{\vc}^{\frac{\fA}{2}} \frac{\sin(\varphi)}{\cos^3(\varphi)}\E_{\Psi}[\pr|\Phi = \varphi]~d\varphi,
\end{align}
where
\begin{align} \nonumber
\E_{\Psi}[\pr&|\Phi = \varphi] = \E_{\Psi}[\pr |\Phi = \varphi, \text{LoS}]\cdot\plos(\varphi) \\ \label{Epsi}
 &+ \E_{\Psi}[\pr |\Phi = \varphi, \text{NLoS}]\cdot\pnlos(\varphi).
\end{align}
Now we have
\begin{align}  \nonumber
\E_{\Psi}[\pr &|\Phi = \varphi, \text{LoS}] = \E_{\Psi}\left[\frac{\Pt}{\Lf \slos}\right] \\ \label{EpsiLoS}
&= \frac{\Pt}{\Lf} \E_{\Psi}\left[\frac{1}{\slos}\right] =  \frac{\Pt}{\Lf} 10^{\frac{-\mlos+v\vlos/2}{10}},
\end{align}
where the last equation follows from the fact that using \eqref{psi_distr} $1/\slos$ adopts a log-normal distribution as
\begin{equation} \nonumber
\ln \left[\frac{1}{\slos}\right] \sim \mathcal{N}\left(-v\mlos,v^2\vlos\right);~~~v=\frac{\ln(10)}{10}.
\end{equation}
Similarly we have
\begin{align} \label{EpsiNLoS}
\E_{\Psi}[\pr &|\Phi = \varphi, \text{NLoS}] = \frac{\Pt}{\Lf} 10^{\frac{-\mnlos+v\vnlos/2}{10}}.
\end{align}
Finally using \eqref{Lf}, \eqref{Plos}, \eqref{mIagg2}, \eqref{Ephi_psi}--\eqref{EpsiNLoS} the desired result is obtained.

\section{Proof of Proposition \ref{PcovApp}} \label{PcovApp_proof}	

To obtain the average LoS probability of the closest drone $\bar{\mathcal{P}}_\text{LoS}^{\hspace{1pt}\text{c}}$ one can write
\begin{equation} \label{plosc-derivation}
\bar{\mathcal{P}}_\text{LoS}^{\hspace{1pt}\text{c}} = \int_0^{\frac{\fA}{2}} \plos(\varphi) f_{\Phi_\text{c}}(\varphi) ~d\varphi 
\end{equation}
where $f_{\Phi_\text{c}}(\varphi)$ is the pdf of the closest drone's location over $\A$ obtained in \eqref{pdf-closest} and $\plos$ is expressed in \eqref{Plos}. Therefore by using these equations in \eqref{plosc-derivation} the desired result is attained.

\end{appendices}

\vspace{10mm}

\bibliographystyle{IEEEtran}

\bibliography{IEEEabrv,Refs}

\begin{thebibliography}{10}
\providecommand{\url}[1]{#1}
\csname url@samestyle\endcsname
\providecommand{\newblock}{\relax}
\providecommand{\bibinfo}[2]{#2}
\providecommand{\BIBentrySTDinterwordspacing}{\spaceskip=0pt\relax}
\providecommand{\BIBentryALTinterwordstretchfactor}{4}
\providecommand{\BIBentryALTinterwordspacing}{\spaceskip=\fontdimen2\font plus
\BIBentryALTinterwordstretchfactor\fontdimen3\font minus
  \fontdimen4\font\relax}
\providecommand{\BIBforeignlanguage}[2]{{%
\expandafter\ifx\csname l@#1\endcsname\relax
\typeout{** WARNING: IEEEtran.bst: No hyphenation pattern has been}%
\typeout{** loaded for the language `#1'. Using the pattern for}%
\typeout{** the default language instead.}%
\else
\language=\csname l@#1\endcsname
\fi
#2}}
\providecommand{\BIBdecl}{\relax}
\BIBdecl

\bibitem{agiwal2016next}
M.~Agiwal, A.~Roy, and N.~Saxena, ``Next generation 5{G} wireless networks: a
  comprehensive survey,'' \emph{{IEEE} Commun. Surveys Tuts.}, vol.~18, no.~3,
  pp. 1617--1655, third~quarter 2016.

\bibitem{miranda2016survey}
K.~Miranda, A.~Molinaro, and T.~Razafindralambo, ``A survey on rapidly
  deployable solutions for post-disaster networks,'' \emph{{IEEE} Commun.
  Mag.}, vol.~54, no.~4, pp. 117--123, Apr. 2016.

\bibitem{kopardekar2014unmanned}
P.~H. Kopardekar, \emph{Unmanned Aerial System (UAS) Traffic Management (UTM):
  Enabling Low-Altitude Airspace and UAS Operations}, 2014.

\bibitem{azari2016optimal}
M.~M. Azari, F.~Rosas, K.-C. Chen, and S.~Pollin, ``Optimal uav positioning for
  terrestrial-aerial communication in presence of fading,'' in \emph{Global
  Communications Conference (GLOBECOM), 2016 IEEE}.\hskip 1em plus 0.5em minus
  0.4em\relax IEEE, 2016, pp. 1--7.

\bibitem{mozaffari2017mobile}
M.~Mozaffari, W.~Saad, M.~Bennis, and M.~Debbah, ``Mobile unmanned aerial
  vehicles (uavs) for energy-efficient internet of things communications,''
  \emph{arXiv preprint arXiv:1703.05401}, 2017.

\bibitem{azari2016gcw}
M.~M. Azari, F.~Rosas, K.-C. Chen, and S.~Pollin, ``Joint sum-rate and power
  gain analysis of an aerial base station,'' in \emph{Global Communications
  Conference (GLOBECOM)}.\hskip 1em plus 0.5em minus 0.4em\relax IEEE, 2016.

\bibitem{azari2017ultra}
------, ``Ultra reliable uav communication using altitude and cooperation
  diversity,'' \emph{arXiv preprint arXiv:1705.02877}, 2017.

\bibitem{mozaffari2016efficient}
M.~Mozaffari, W.~Saad, M.~Bennis, and M.~Debbah, ``Efficient deployment of
  multiple unmanned aerial vehicles for optimal wireless coverage,'' \emph{IEEE
  Communications Letters}, vol.~20, no.~8, pp. 1647--1650, 2016.

\bibitem{hayajneh2016drone}
A.~M. Hayajneh, S.~A.~R. Zaidi, D.~C. McLernon, and M.~Ghogho, ``Drone
  empowered small cellular disaster recovery networks for resilient smart
  cities,'' in \emph{Sensing, Communication and Networking (SECON Workshops),
  2016 IEEE International Conference on}.\hskip 1em plus 0.5em minus
  0.4em\relax IEEE, 2016, pp. 1--6.

\bibitem{fotouhi2016dynamic}
A.~Fotouhi, M.~Ding, and M.~Hassan, ``Dynamic base station repositioning to
  improve performance of drone small cells,'' in \emph{Globecom Workshops (GC
  Wkshps), 2016 IEEE}.\hskip 1em plus 0.5em minus 0.4em\relax IEEE, 2016, pp.
  1--6.

\bibitem{fotouhi2017dynamic}
------, ``Dynamic base station repositioning to improve spectral efficiency of
  drone small cells,'' \emph{arXiv preprint arXiv:1704.01244}, 2017.

\bibitem{chandrasekharan2016designing}
S.~Chandrasekharan, K.~Gomez, A.~Al-Hourani, S.~Kandeepan, T.~Rasheed,
  L.~Goratti, L.~Reynaud, D.~Grace, I.~Bucaille, T.~Wirth \emph{et~al.},
  ``Designing and implementing future aerial communication networks,''
  \emph{{IEEE} Commun. Mag.}, vol.~54, no.~5, pp. 26--34, May 2016.

\bibitem{al2014modeling}
A.~Al-Hourani, S.~Kandeepan, and A.~Jamalipour, ``Modeling air-to-ground path
  loss for low altitude platforms in urban environments,'' in \emph{Global
  Communications Conference (GLOBECOM)}.\hskip 1em plus 0.5em minus 0.4em\relax
  IEEE, 2014, pp. 2898--2904.

\bibitem{elsawy2016modeling}
H.~ElSawy, A.~Sultan-Salem, M.-S. Alouini, and M.~Z. Win, ``Modeling and
  analysis of cellular networks using stochastic geometry: A tutorial,''
  \emph{{IEEE} Commun. Surveys Tuts.}, vol.~19, no.~1, pp. 167--203, 2016.

\bibitem{andrews2011tractable}
J.~G. Andrews, F.~Baccelli, and R.~K. Ganti, ``A tractable approach to coverage
  and rate in cellular networks,'' \emph{IEEE Transactions on Communications},
  vol.~59, no.~11, pp. 3122--3134, 2011.

\bibitem{naghshin2017coverage}
V.~Naghshin, M.~C. Reed, and N.~Aboutorab, ``Coverage analysis of packet
  multi-tier networks with asynchronous slots,'' \emph{IEEE Transactions on
  Communications}, vol.~65, no.~1, pp. 200--215, 2017.

\end{thebibliography}

\vfill

\end{document}